\documentclass[11pt]{article}

\usepackage{amsfonts,amssymb,amsmath,latexsym,xfrac,ae,aecompl}

\usepackage{graphicx}

\newcommand{\FFF}{\vspace*{\bigskipamount}}

\newcommand{\cO}{\mathcal{O}}

\newcommand{\mF}{\mathbb{F}}

\newcommand{\remove}[1]{}

\newlength{\pagewidth}
\newlength{\captionwidth}
\newcommand{\qed}{\hfill $\square$ \smallbreak}
\newenvironment{proof}{\noindent{\bf Proof:}}{\qed}

\newtheorem{theorem}{Theorem}
\newtheorem{lemma}{Lemma}

\newlength{\boxwidth}
\begin{document}   

\baselineskip    3ex
\parskip         1ex

\title{Deterministic Computations on a PRAM\\ 
with Static Processor and Memory Faults~\footnotemark[1]\FFF\FFF}

\author{Bogdan S. Chlebus~\footnotemark[2] \and
		Leszek G{\k a}sieniec~\footnotemark[3] \and 
		Andrzej Pelc~\footnotemark[4]}

\date{}

\maketitle

\footnotetext[1]{
This work was published as~\cite{ChlebusGP-FI03}.
The results of this paper appeared in a preliminary form in~\cite{ChlebusGP95}.}

\footnotetext[2]{Department of Computer Science and Engineering,
University of Colorado Denver,
Denver, CO 80217, USA.
Work supported by the National Science Foundation under Grant No. 0310503.}

\footnotetext[3]{Department of Computer Science,
University of Liverpool, 
Liverpool L69 3BX, United Kingdom.
Work supported by the EPSRC grant GR/N09855/01.}

\footnotetext[4]{D\'{e}partement d'informatique et d'ing\'enierie, 
Universit\'{e} du Qu\'{e}bec en Outaouais,
Gatineau, Qu\'{e}bec J8X 3X7, Canada.
Work supported by the NSERC grant 0008136.}

\thispagestyle{empty}

\vfill

\begin{abstract}
We consider Parallel Random Access Machine (PRAM) which has some 
processors and memory cells faulty.
The faults considered are static, i.e., once the machine starts to operate, 
the operational/faulty status of PRAM components does not change.
We develop a deterministic simulation of a fully operational PRAM
on a similar faulty machine which has constant fractions of
faults among processors and memory cells.
The simulating PRAM has $n$ processors and $m$ memory cells,
and simulates a PRAM with $n$ processors and a constant
fraction of $m$ memory cells.
The simulation is in two phases: it starts with  preprocessing, which is 
followed by the simulation proper performed in a step-by-step fashion.
Preprocessing is performed in time $\cO((\frac{m}{n}+ \log n)\log n)$.
The slowdown of a step-by-step part of the simulation is $\cO(\log m)$.
\end{abstract}

\vfill

\setcounter{page}{1}

\newpage

\section{Introduction}

Computers are prone to hardware failures, and this vulnerability 
increases with the growth of the number of processing elements 
in multiprocessor machines.
It is a natural problem to investigate how to execute algorithms on 
machines with some faulty components.
The most satisfactory solution is to design a simulation mechanism
which is able to adapt automatically any executed program to the 
encountered pattern of faults.

We present a paradigm of deterministic simulation of
parallel algorithms in a synchronous shared-memory environment with faults.
The developed simulation is performed on a PRAM with $n$ 
processors, some possibly faulty, and $m$~memory cells, some possibly 
faulty too.
This simulating machine is called a {\em simulator}.
The simulated machine is a fully operational PRAM with $n$~processors
and a constant fraction of $m$~memory cells.
We assume that there are constant fractions of the number $n$ of processors 
and the number $m$ of memory cells of the simulator which are upper
bounds on the number of faulty processors and memory cells, respectively.
The distribution of faulty components may be arbitrary, subject only to
these bounds on the number of faults. 

It is for the first time that a deterministic simulation technique 
is designed for a PRAM under this model of fault-prone 
processors and fault-prone shared-memory cells. 
Note that the ability of a PRAM with some faulty memory cells to 
perform deterministic and fast computations under the worst-case 
fault distribution is not obvious.
Namely, if the simulation is deterministic then the sequences of memory 
cells accessed by processors are predetermined until some
operational cells are encountered.
This implies that some operational processors may {\em never\/} (within the
time bounds of the supposed fast simulation) find fault-free memory cells, 
and hence may never be able to communicate with other processors.
To circumvent this obstacle, the developed simulation is designed 
to rely on a constant fraction of operational processors.
In particular, we face the problem how to retrieve the input 
from the information available to the processors that are active 
in the simulation.
To this end we  use the method of information dispersal.

The term ``simulation'' may be used with various meanings in the realm of
distributed or parallel computing.
One of them concerns specifications of algorithms in layered systems,
in which complex solutions are constructed from simpler primitives; see
e.g.~\cite{Attiya-Welch-book2004}.
A simulation in this sense requires a precise development of the underlying
model of computation.
In the context of PRAM algorithms, the term simulation typically denotes 
a way of interpreting PRAM algorithms, which are written in a syntax 
allowing direct access to shared memory, in message-passing 
or distributed-memory environments, see e.g. \cite{Harris94}.
This requires developing an implementation of shared-memory access, 
and the emphasis is on an asymptotic performance of the underlying
communication mechanism.
The meaning that we are using is closer to the latter one.
We show that it is possible to run PRAM algorithms, which have been 
designed for an ideal fully reliable PRAM, on a PRAM with faults, 
achieving a moderate slowdown.

\noindent
{\bf Overview of the simulation.}
The simulation is in two phases: it starts with setting up the machine,
which is then followed by a phase in which an arbitrary program
is executed in a step-by-step fashion.
After preprocessing has been completed, some of the operational processors 
are identified as suitable to participate in the simulation. 
Then they are given new identification numbers to be used during simulation.
Similarly, some operational memory cells are given new virtual
addresses, and there is a mechanism created to access them by these new
coordinates.
More details are as follows.

\noindent
{\sc Preprocessing:}
\begin{quote}
A constant fraction of $n$ processors organize a constant fraction of 
$m$ memory cells of the shared memory into data structures to make possible 
accessing the contents of memory cells of the simulated machine 
by their addresses in the simulated memory in time $\cO(\log m)$ per access.
First the $m$ memory cells are conceptually partitioned into blocks of size
$m/n$ and assigned to all the processors, including the faulty ones.
Then each  operational processor verifies if it has available at least 
a certain constant number of operational cells within a certain 
constant-size range, 
and if this is the case then it considers itself {\em active}.
The remaining {\em passive\/} operational processors remain idle since this 
moment.
The active processors organize the memory available to them into a balanced 
search tree called a memory tree.
This tree is sufficient to provide a fast access to virtual memory 
if concurrent-read primitive is available, like in CREW or CRCW PRAMs.
A solution in the case of the EREW variant is more complex: 
during preprocessing, the processors organize themselves into a balanced 
binary tree, called processor tree, which is then used to facilitate 
concurrent access to the virtual memory without concurrent reading.
This part of preprocessing is performed in time 
$\cO((\frac{m}{n}+ \log n)\log n)$.

An input is provided to all the processors in a way which does not
depend on the distribution of faults.
Each processor, whether operational or not, obtains its part of the input,
which is of size of a constant number of memory cells per processor.
We use redundancy as a safeguard against some part of the input being
lost due to processor or memory faults.
The input is provided encoded according to a variant of 
information dispersal.
Active processors communicate with one another and retrieve the whole
original input from the information available to them.
It takes time $\cO(\log^2 n)$ to obtain the whole input.
\end{quote}

\noindent
{\sc The simulation proper:}
\begin{quote}
A program to execute and data on which it is to be run are provided as input.
This program is simulated in a step-by-step fashion.
A single active processor simulates $\cO(1)$ processors of the 
simulated machine.
The crucial part is performing the operation of accessing a shared memory 
cell, for either reading or writing.
The goal is to locate the physical address of a cell in the memory
of the simulator, given a virtual address of memory of the 
simulated machine. 
In the case of the EREW PRAM, each operation of memory access 
requires an ad hoc tree of processors, called an organized tree, 
which is constructed on top of the structure of the processor tree,
and which is discarded after the memory-access operation has been completed.
The slowdown of the proper part of the simulation is asymptotically 
the same as the time to perform a single memory-cell access and is 
$\cO(\log m)$.
\end{quote}

\noindent
{\bf Related work.}
A lot of research has been done to develop simulation techniques
transforming PRAM algorithms, which have been designed to operate 
in an ideal fault-free synchronous environment, into algorithms that are 
reliable when run on fault-prone or asynchronous machines.
Various approaches are possible, depending on the nature of faults
(static versus dynamic, deterministic versus stochastic, fail-stop
versus restartable), the efficiency criteria (time versus work), the 
properties and capabilities of the underlying model (synchronous versus 
asynchronous, concurrent read and write versus exclusive read and write), 
or the types of simulations (randomized versus deterministic).
In the static model of failures, a computation starts with some
components out of order, and the status of all the components is
not changed in the course of an execution.
A relatively mild form of dynamic failures is that of processor
crashes.
Models of failures may also allow delays or restarts of failed processors, 
such behaviour can be interpreted as a form of asynchrony.
Asynchrony in PRAMs makes the model closer to the way actual 
parallel systems operate (\cite{Valiant90-CACM,Valiant90-handbook}).
For results about the PRAM model in its asynchronous version 
see~\cite{AspnesH90,ColeZ89,Gibbons89,Herlihy91}.
As far as the number of faults is concerned, one may require that
all but one processors are allowed to fail.
To have a meaningful cost measure for such a situation, it is natural to
consider {\em work,} which is sometimes called the 
{\em available processor steps} (see~\cite{KanellakisS97-book}).
Under this cost measure any operational processor involved in the computation
contributes a unit of work for any clock cycle during which it is operational, 
even if it is idle.

Kanellakis and Shvartsman~\cite{KanellakisS97-book,KanellakisS92} identified 
a generic algorithmic problem called \emph{Write-All}:
given an array initialized with zeroes, the task is to change 
the values stored at all the locations in the array.
A solution of \emph{Write-All} on a PRAM with faults can be used iteratively 
to convert any shared-memory computation into a computation resilient 
to processor failures (see \cite{KanellakisS97-book,KedemPS90,Shvartsman91}).
The \emph{Write-All} problem was studied in a series of papers.
For the purpose of giving examples of specific performance bounds 
in this overview, we consider a special case when both the number 
of processors and the size of the array are equal, let $n$ be this number.
For deterministic computations which are robust against fail-stop failures, 
the best known algorithm solving \emph{Write-All} performs work 
$\cO(n\log^2 n/\log\log n)$, it was given in~\cite{KanellakisS92}.
If restarts may happen than the most efficient known algorithm was proposed
by Anderson and Woll~\cite{AndersonW97}, it performs work $\cO(n^{1+\epsilon})$, 
for arbitrary positive~$\epsilon$. 
This algorithm is based on an extension of the idea of a (binary) 
{\em progress tree\/} proposed first by Buss, Kanellakis, 
Ragde and Shvartsman~\cite{BussKRS96}, 
it uses $q$-ary rather than binary tree,  for suitable $q$.
The algorithm is non-constructive in that it relies on families 
of permutations with certain good properties, which are only known to exist;
the best known constructive families with relevant properties
are given in~\cite{NaorR95}.
Motivated by the algorithm in~\cite{AndersonW97}, Kanellakis and 
Shvartsman~\cite{KanellakisS97-book} proposed a constructive algorithm. 
The performance of this algorithms has not been fully analyzed yet, 
see~\cite{ChlebusDKMSV-SPAA01} for a work in this direction.
The expected work $\cO(n\log n)$ of randomized computations can be 
achieved even with restarts, as was shown by Martel et al.~\cite{MartelPS92}.
A logarithmic-time lower bound for randomized solutions of \emph{Write-All} showed 
by Martel and Subramonian~\cite{MartelS94} is in contrast with 
the algorithm of constant amortized expected slowdown developed 
by Kedem et al. in~\cite{KedemPRS91}, 
for a stochastic model of faults.
Groote et al.~\cite{GrooteHMV01} gave an algorithm for \emph{Write-All} which is efficient for the practical applications when the size of the array is significantly larger than the number of processors.
A similar problem \emph{Do-All} for a faulty message-passing environment
has been proposed by Dwork et al.~\cite{DworkHW98}: 
there are a number of similar and independent tasks which need 
to be performed in any order.

Chlebus et al.~\cite{ChlebusGI94} initialized the study of PRAM 
with memory faults.
They considered CRCW PRAM with deterministic memory faults, and designed 
efficient randomized PRAM simulations for both static and dynamic errors.
Indyk~\cite{Indyk96} studied computations exploiting bit operations and
resilient to memory faults.
G{\k a}sieniec and Indyk~\cite{GasieniecI97} studied PRAM with static memory 
faults and developed both EREW and CRCW simulations.
Our simulation is related to the results obtained in \cite{ChlebusGI94, DiksP00, GasieniecI97, GasieniecP96,Indyk96}.
Two  simulations presented in \cite{ChlebusGI94} are for the static
memory errors: one on the optimal number 
of $n/\log n$ processors and with  slowdown $\cO(\log n)$,
the other operating in real time on $n \log n$ processors. 
The simulations developed in \cite{ChlebusGI94} are
randomized, the performance bounds are expected, and the model 
is a  CRCW PRAM, whereas the simulation presented in this paper is 
deterministic and applicable to the whole PRAM family,
in particular to the weakest  EREW PRAM.
The preprocessing of the presented deterministic simulation developed
in this paper includes building a list of active processors 
and then a tree-like data structure. 
If $m=\Theta(n)$ then no trees are needed, just a list of active processors
suffices, as was shown in the conference version of this paper.
It has been shown that the list-construction part can be designed
to run within time $\cO(\log n)$, see \cite{DiksP00,GasieniecI97}.
Diks and Pelc~\cite{DiksP96} developed algorithms for the EREW PRAM, 
which are reliable and efficient under a stochastic processor-fault model.

Fault-tolerance of other models of parallel computing has also been studied. 
We mention briefly some of the models most closely related to this research.
Berenbrink et al.~\cite{BerenbrinkHS96} studied a distributed 
memory machine with randomly failing memory modules, and presented fast 
PRAM simulations in this model.
Chlebus et al.~\cite{ChlebusGI96} presented randomized simulations
of a PRAM on a  distributed memory machine with faulty memory cells, under
arbitrary distribution of faults in the memory modules.
Kontogiannis et al.~\cite{KontogiannisPSY00} studied simulations 
of the operational BSP on a BSP with processor faults.

Failure models are often described in terms of adversaries, who are
responsible for incurring worst-case scenarios.
A situation when the faults are set at the beginning and then do not change
is called {\em static}, it is the one considered in this paper.
A more complex situation is modeled by the {\em oblivious\/} adversarial 
model, in which the timing of faults is determined before the algorithm starts.
The dynamic case, when faults may occur in the course of computations at times
decided online by an adversary, are captured by the {\em adaptive\/} 
adversarial model.
If the number of faults is required to be upper bounded by a constant 
fraction of the total number of given components, then the adversary is said 
to be {\em linearly bounded}, it is the one considered in this paper.
Adaptive adversaries, with an upper bound on the number of faults, 
are in two variants. 
{\em Weakly-adaptive\/} adversaries designate fault-prone components, of
a prescribed number, prior to the start of algorithm 
and may fail only them in the course of an algorithm, at arbitrary steps.
{\em Strongly-adaptive\/} adversaries may decide on faults fully online.
The distinction between strongly and weakly adaptive models is meaningful only
if algorithms are randomized, because otherwise the adversary can predict
the future behaviour of algorithms.
The following are examples of applications of these adversarial models found in
recent literature.
Linearly-bounded adversaries were considered in~\cite{DiksP00} 
in the context of deterministic broadcasting and fault diagnosis.
A stochastic variant of a linearly-bounded adversary controlling processor 
faults in the BSP model was studied in~\cite{KontogiannisPSY00}.
Weakly-adaptive linearly-bounded adversaries were considered 
in~\cite{ChlebusGI94} in the context of faulty CRCW PRAM
with faulty memory cells.

Information dispersal is a term used for any method to encode a string of 
bits into $x$ pieces in such a way that the original
string can be retrieved from any $y$ of the pieces.
If the original string consists of $\ell$ bits, then each of the pieces
is to be of size $\ell/y$.
The numbers $y\le x$ are parameters, they may be arbitrary positive integers.
Known efficient implementations of information dispersal use algebraic methods,
often related to the theory of erasure and error-correcting codes.
The Rabin's method~\cite{Rabin89} is based on the properties of the Cauchy matrix
$[\frac{1}{a_i+b_i}]_{i,j}$.
Probably a conceptually simplest solution is based on the fact that a
polynomial of degree $k$ is uniquely determined by its values at $k+1$ points,
a similar idea underlies the Reed-Solomon error-correcting codes~\cite{MacwilliamsS77-book}.
It has been proposed by Preparata~\cite{Preparata89} and Lyuu~\cite{Lyuu93-book},
but has been anticipated much earlier by McEliece and Sarwate~\cite{McElieceS81}.
An advantage of this approach is that encoding and decoding can be done
efficiently by methods based on the Fast Fourier Transform.
The paradigm of information dispersal has been used in 
routing on faulty networks in a pioneering work of Rabin~\cite{Rabin89}, 
in cryptography by Shamir~\cite{Shamir79} and Ben-Or et al.~\cite{Ben-OrGW88}, and in PRAM simulation on faulty networks by Aumann and Ben-Or~\cite{AumannB91}.
As far as  unreliable PRAMs are concerned, the method of information 
dispersal has been previously used by Aumann et al.~\cite{AumannKPR93} in a context when memory cell values could be unreachable.
One can consult a book of Lyuu~\cite{Lyuu93-book} for an exposition of 
known implementations of information dispersal, a comprehensive account 
of its applications in parallel and distributed computing, and an extensive
bibliography.
We apply the information dispersal method to retrieve the whole input from the 
parts of input data available to processors active in the simulation.
Our approach is that of Reed-Solomon codes, and the exposition 
is self-contained as far as understanding the basic machinery and 
complexity estimates is concerned.

The paper is organized as follows.
Section~\ref{model} deals with the model, and in particular defines 
precisely the type of faults. 
In Section~\ref{preprocessing} the stage of preprocessing is described.
Section~\ref{simulation} gives details of the proper part of the simulation.

\section{Model of Computation}
\label{model}

This section defines the model of computation and discusses 
the assumptions we make concerning the nature, distribution 
and number of faults.

The model of computation that we consider in this paper is
the {\em Parallel Random-Access Machine\/} (PRAM),  see e.g.~\cite{JaJa92-book}.
PRAM  is a formal model that facilitates design and exposition 
of parallel algorithms.
We consider a synchronous variant.
The simulator consists of a shared memory of 
$m$ memory cells and of $n$  
processors which operate in synchronized steps. 
Each processor has its own local control and its own local memory.
It knows its identification, which is a unique number assigned to it in the
range from $1$ through $n$, 
we refer to it as the {\it physical index} of the processor.
The shared memory is assumed to consist of individually addressable
memory cells.
Each processor can access arbitrary locations, either in the shared 
memory or in its  local memory, in a single step. 
PRAM can be efficiently simulated on distributed-memory machines,
see \cite{Harris94,Heide92,Valiant90-handbook} for overviews of PRAM-simulation techniques.

PRAM can be considered as a direct generalization of the sequential model
RAM, in the sense that if we select one processor and forget
about the remaining ones then the resulting machine is a RAM.
Usually RAM is defined as a machine with unlimited amount of memory.
However, in the PRAM case, the size of memory can be a factor that determines
the speed of computation, because processors may use the shared memory 
to communicate with one another, and the size of memory limits the bandwidth of
such a communication. 
This issue becomes even more important in the context of a PRAM simulation on 
a PRAM with possibly faulty memory, because some quantitative assumptions need 
to be made regarding the distribution of faults.
In this paper when we refer to a PRAM then it is always a model with a
specific number of processors and a specific  size of shared 
memory, and these parameters denote the total numbers of these components,
with some of them possibly faulty. 

The PRAM model has variants categorized according to 
the semantics of concurrent access of processors to the shared memory
(\cite{JaJa92-book}).
The weakest {\sl exclusive-read exclusive-write}\/ (EREW)
allows only at most one processor to read or write to any location
in the shared memory during each step. 
For the EREW PRAM, violating the exclusivity restrictions 
results in a runtime error.
The {\sl concurrent-read exclusive-write}\/ (CREW) allows many  
processors to read from the same location in the shared memory 
at the same time, but simultaneous writes are not allowed.
The strongest {\sl concurrent-read concurrent-write}\/ (CRCW)
allows for simultaneous reads and writes by many processors to 
the same locations in the shared memory. 
In this paper we develop a universal simulation method that can be applied to
any variant of a PRAM with faults.
The simulated machine is of the same kind as the simulator, that is,
EREW, CREW or CRCW, respectively.
Our exposition concentrates on the case of EREW PRAM, 
which is most restrictive and hence requires more effort than the other
variants.

We study such a PRAM in which some among processors and some among 
shared memory cells may be faulty at the moment when computation starts.
The kinds of faults we are interested in are usually called
deterministic and static.
Here {\em deterministic\/} means that there are no stochastic assumptions
concerning the distribution of faults, and that the performance of the
simulation is measured according to the worst case, that is, with respect 
to the worst possible distribution of faults. 
The term {\em static\/} means that the  faulty/operational status of 
error-prone components remains the same in the course of a computation 
as it was when the simulation started.
Processor faults are of the {\em fail-stop\/} kind, i.e., 
faulty processors do not perform any action during the computation.
Each processor has a local storage with a constant number of memory cells,
which we assume to be  sufficiently large for the
purpose of the given deterministic simulation.
Local memories are assumed not to be faulty.
An important property of the model is the capability of 
processors to recognize faulty cells in shared memory.
Namely, whenever a processor attempts to access a memory 
cell then it is immediately notified by hardware whether the cell 
is operational or not.
During the setup phase, each shared memory cell of the simulated machine 
is assigned a shared memory cell of the simulator, 
for the purpose of storing its contents.
Hence the memory cells of the simulated machine are exactly as those 
of the simulator, in terms of the number of bits they can store.

We assume that at least a certain constant fraction of $n$~processors and 
at least a certain constant fraction of $m$ memory cells are operational. 
If we are to be able to cope with an arbitrary location of faults
by a simulation which is both deterministic and 
based on a fixed initial assignment of blocks of memory to processors,
then additional assumptions about the number of faults are needed.
To see this, consider a PRAM with $n$ processors and $n$
memory cells, such that up to half of the processors and up to half of 
memory cells may fail; then, for any fixed assignment of cells to processors, 
it might happen that the operational processors are assigned to exactly 
the faulty cells.
We make suitable assumptions concerning the numbers of faults
to avoid such a situation.

We assume that there are two constants $0<f_s<1$ and $0<f_p<1$ such that the 
number of faulty processors is at most $f_p n$, and the number of 
faulty memory cells is at most $f_s m$.
An additional assumption about $f_p$ and $f_s$ is that $f_p+f_s<1$.
A contiguous segment of memory cells is an 
{\it $\alpha$-$\beta$-good segment}, 
for integers $\alpha\ge \beta>1$, if it is of size at most 
$\alpha$, it begins and ends at operational cells, and contains exactly 
$\beta$ operational memory cells.
An $\alpha$-$\beta$-good segment is uniquely specified by the address of 
its first operational cell, we call it this {\em segment's address}.
The shared memory is conceptually partitioned into $n$ contiguous segments
of memory cells of size $\lfloor \frac{m}{n}\rfloor$, called {\em blocks}.
The $i$th block is assigned to the $i$th processor.
A processor is {\em $\alpha$-$\beta$-active\/} if it is both operational 
and its assigned block contains at least one $\alpha$-$\beta$-good segment, 
otherwise it is said to be {\em $\alpha$-$\beta$-dormant}.
The block assigned to an active processor is also called active.

$\alpha$-$\beta$-good segments are used as records to store information.
The memory is organized into data structures during the setup phase of
simulation.
Number $\beta$ is required to be large enough so that an $\alpha$-$\beta$-good
segment has sufficient space to store all the pointers needed. 
The size of number $\beta$ is determined implicitly by the simulation
mechanism, as described in Sections~\ref{preprocessing} 
and~\ref{simulation}.
 From now on we assume $\beta$ to be a sufficiently large fixed constant. 
Parameter $\alpha$ depends on all the constants $f_p$, $f_s$ and~$\beta$.

\begin{lemma}
\label{l1}
For any constants $f_p>0$, $f_s>0$, such that $f_p+f_s<1$, 
and integer $\beta>1$, 
there is an integer $\alpha\ge \beta$ such that if a faulty PRAM 
with $n$ processors has at least $m\ge \alpha n$ shared memory cells 
then the number of $\alpha$-$\beta$-active processors is at least 
$n\frac{1-(f_p+f_s)}{2}$.
\end{lemma}
\begin{proof}
The assumption $m\ge\alpha n$ guarantees that each processor has assigned at
least $\alpha$ shared memory cells.
Let us consider an arbitrary pair $\alpha\ge\beta$, 
and let $d$ be the corresponding number of 
$\alpha$-$\beta$-dormant processors.
At most $f_p n$ processors are faulty.
If an operational processor is $\alpha$-$\beta$-dormant then each segment 
of size $\alpha$ in its memory contains at least $\alpha-\beta+1$ faulty cells.
Hence the total number $x$ of faulty memory cells satisfies the inequalities
\[
(d-f_p n)\left\lfloor\frac{m}{n}\right\rfloor 
\frac{\alpha-\beta+1}{\alpha} \le x \le f_s m\ .
\]
By removing the floor function we obtain a weaker inequality
\[
(d-f_p n)\left(\frac{m}{n}-1\right) \frac{\alpha-\beta+1}{\alpha}
\le f_s m\ .
\]
Multiplying both sides by $\frac{ \alpha}{\alpha-\beta+1}\frac{n}{m-n}$
and rearranging yields the inequality
\[
d
\le
n \Bigl(f_p+f_s \frac{\alpha}{\alpha-\beta+1}\frac{m}{m-n}\Bigr)
\ .
\]
Since both $\frac{\alpha}{\alpha-\beta+1}$ and $\frac{m}{m-n}$ 
approach~$1$ as $\alpha$ grows to infinity, and $f_p+f_s<1$, the following
inequality 
\[
f_p+f_s \frac{\alpha}{\alpha-\beta+1}\frac{m}{m-n}
<\frac{1}{2}(f_p+f_s+1)
\]
holds for sufficiently large $\alpha$.
Then the number of $\alpha$-$\beta$-active processors is at least 
\[
n\Bigl(1-\frac{f_p+f_s+1}{2}\Bigr)= n\frac{1-(f_p+f_s)}{2}\ ,
\]
which is the estimate we seek.
\end{proof}
$\alpha$-$\beta$-good segments in a block can be located by a greedy 
procedure which scans consecutive cells of the block, 
and, as soon as $\beta$ operational cells are found in a range of at 
most $\alpha$ cells, it designates them as a separate 
$\alpha$-$\beta$-good segment (see the description of Stage~1 of the
preprocessing in Section~\ref{simulation}).
This is done in each block separately.
 From now on, when referring to $\alpha$-$\beta$-good segments, we consider
them as produced in this way.

\begin{lemma}
\label{l2}
If the inequality $\alpha>\frac{\beta-1}{1-f_s}$ holds then there is a constant
$\delta>0$, which depends on $f_s$, $\alpha$ and $\beta$, such that 
the total number of $\alpha$-$\beta$-good segments in all the blocks 
is at least $\delta m$.
\end{lemma}
\begin{proof}
For each operational cell~$w$, consider the contiguous segment $I_w$ of 
$\alpha$ cells starting at~$w$.
If $I_w$ contains at least $\beta$ operational cells then $w$ belongs to an
$\alpha$-$\beta$-good segment.
If $w$ does not belong to any $\alpha$-$\beta$-good segment then $I_w$ 
contains at least $\alpha-1-(\beta-2)=\alpha-\beta+1$ faulty cells,
we say that these faulty cells are {\em bound\/} to~$w$.
A faulty cell can be bound to at most $\beta-1$ operational cells, because
otherwise these operational cells would be within a segment of size at 
most~$\alpha$.
It follows that the number of faulty cells per each operational one that is not
in a $\alpha$-$\beta$-good segment, is at least 
$\frac{\alpha-\beta+1}{\beta-1}$.
Let $A$ be the number of operational cells that do not belong to
$\alpha$-$\beta$-good segments.
The total number of faulty cells is at least 
$A \frac{\alpha-\beta+1}{\beta-1}$.
We have 
\[
A\le \frac{\beta-1}{\alpha-\beta+1} f_s m\ ,
\]
because otherwise the total number of faulty cells would be at least
\[
A \frac{\alpha-\beta+1}{\beta-1}>f_s m\ .
\]
Since the total number of operational cells is at least $(1-f_s) m$,
we obtain that the number of operational cells in $\alpha$-$\beta$-good 
segments is at least 
$\bigl(1-f_s-\frac{\beta-1}{\alpha-\beta+1} f_s\bigr) m$.
One can check directly that if $\alpha>\frac{\beta-1}{1-f_s}$ then 
the following inequality holds:
\[
1-f_s-\frac{\beta-1}{\alpha-\beta+1} f_s>0\ .
\]
Hence if $\alpha>\frac{\beta-1}{1-f_s}$ then the number of 
$\alpha$-$\beta$-good segments is at least $\delta m$, 
where 
\[
\delta=\frac{1}{\alpha} \Bigl[1-\bigl(1-\frac{\beta-1}{\alpha-\beta+1}\bigr) 
f_s\Bigr]\ ,
\]
as was to be shown.
\end{proof}

Our simulation is defined for faulty PRAMs in which the numbers of processors,
memory cells and allowed faults satisfy certain dependencies, defined by
parameters fixed as constants, we call such PRAMs normal.
A precise definition is as follows:

Let $f_p>0$ and $f_s>0$ be two constants such that additionally $f_p+f_s<1$.
Let $\beta>1$ be an integer constant, its value needs to be sufficiently 
large to make the simulation mechanism work as described in detail in 
the following sections.
Let $\alpha>\max\{\beta,\frac{\beta-1}{1-f_s}\}$ be sufficiently large,
such that the conclusion of Lemma~\ref{l1} holds.
Let $\delta$ be the number proved to exist in Lemma~\ref{l2}.
A faulty PRAM with $n$ processors and $m$ shared memory cells 
is said to be {\em normal\/} if it has the following properties:
\begin{enumerate}
\item
The number $m$  satisfies the inequality~$m\ge \max\{\alpha n,1/\delta\}$.
\item
The number of faulty processors is at most $f_p n$.
\item
The number of faulty memory cells is at most $f_s m$.
\end{enumerate}

 From now on, when referring to a faulty PRAM, we mean a specific normal PRAM,
with all the values $f_p,f_s,\alpha,\beta$ fixed.
In view of Lemmas \ref{l1} and \ref{l2}, normal PRAMs have at least 
a constant fraction of $n$ of active processors 
and at least a constant fraction of $m$ of good segments.
The simulation we develop is for a normal PRAM, and the numeric values
of the constants $f_p,f_s,\alpha,\beta$ are used in the simulation. 
They are assumed to be a part of code of the simulating algorithm that 
is run by each processor, in particular they are known by the processors.
In the following sections we consider the constants $\alpha$ and 
$\beta$ as fixed and write shortly {\em good}, {\it active\/} and 
{\em dormant}, instead of long-winded phrases $\alpha$-$\beta$-good,  
$\alpha$-$\beta$-active, and $\alpha$-$\beta$-dormant, respectively.

The details of input/output operations for a PRAM have never been much
discussed, due mainly to the fact that the model ignores the real costs 
of communication between processors and memory modules.
Usually it is simply assumed that the input has been stored,
prior to the start of a computation,
either in the processors' local memories or in the shared memory.
Our approach is that the input is provided directly to the processors, 
since the shared memory is not reliable.
It would be hardly acceptable to assume that input operations start 
{\em after\/} the preprocessing phase has been completed.
In practice that would require a very flexible and complicated 
input hardware, since categorizing of processors into 
active and dormant is done during preprocessing, and  a processor may 
stay idle  during the computation either if it is faulty 
or its assigned segment of shared memory cells does not contain 
sufficiently many fault-free elements.
We assume that the input is provided uniformly to all the processors, 
whether faulty or not, {\em before\/} preprocessing begins.
Then operational processors store it in their local memories, 
and preprocessing starts.

\section{Preprocessing}
\label{preprocessing}

This section deals with setting up
a machine for the proper part of simulation.
We begin with a description of what preprocessing achieves.

Consider a normal faulty PRAM with $n$ processors and $m$ memory cells.
After it has completed preprocessing, the simulated machine has 
a range from $1$ through $m'$ of shared memory addresses available for 
computation, where $m'$ is a constant fraction of $m$, which depends on all 
the constants involved in the definition of a normal PRAM.
We refer to these addresses as {\em virtual} ones.
The physical cells storing the contents of virtual cells are
scattered throughout the active blocks of memory.
To have a small overhead of simulation, in terms of the time needed to 
locate physical addresses of the memory cells assigned to given virtual
addresses, both active processors and good segments are organized as
balanced binary trees.

\noindent
{\sc 1. Processors:}
Each operational processor knows if it is active or not.
All the active processors are ordered by their physical indices, 
and each processor knows its rank in this ordering.
We refer to the $i$th active processor as~$P_i$.
The number~$i$ is said to be the {\em active index\/} of processor~$P_i$,
and also the active index of its assigned block.
The active index should be distinguished from the corresponding physical 
one.
Active processors are organized into a complete binary tree. 
This is done by storing pointers at memory cells 
designated for this purpose by the processors.

\noindent
{\sc 2. Memory:}
Each active block has a list of disjoint good segments.
The first among them, that is, with the smallest physical address, 
is called the {\em contact segment}, 
and its address is the {\em contact address\/} of this block.
The remaining good segments are {\em regular\/} ones.
Each operational memory cell of a good segment has a purpose assigned to
it, which may depend on the kind of the segment.
The role of a cell may be either {\em structural}, when it stores 
information about the organization of memory of the simulator, for
instance a pointer value in a data structure, or it may provide
virtual storage, when it stores the content of some 
shared memory address of the simulated machine.
The specification of simulation we give determines implicitly how many 
structural cells are required in a good segment, and what they are to store, 
the remaining operational cells can be used for virtual storage.
We require the parameter $\beta$ to be so large that each good segment 
has at least one operational cell apart from the structural ones.

All good segments are organized as a complete binary tree called the
{\em memory tree}.
The in-ordering of segments in the memory tree
is the same as the ordering by their physical addresses,
that is, the addresses of nodes in the
left/right subtree of any node are smaller/larger than the address of 
this node.

All the contact segments are organized as a doubly-linked list.
Since there is a one-to-one correspondence between these segments and both the
active processors and active blocks, this list may be considered as both
a list of active processors and of such blocks.
The ordering of this list is exactly the same as that of the active
processors by their indices.

Each contact segment has a designated pair of records, which are used to 
build a complete binary tree, called the {\em processor tree}.
The structure of the tree is such that typically one of the records 
is a leaf of this tree, the other one is an internal node.
The ordering of the leaves of the processor tree according to the 
in-ordering is the same as their ordering by the indices 
of the corresponding processors.

We start a description of preprocessing by summarizing it 
as a sequence of stages on Figure~\ref{fig-1}.
Some of the terms used in this summary are defined later on 
in a detailed presentation of the stages.

%	figure: pseudocode of preprocessing

\begin{figure}[t]
\rule{\textwidth}{0.75pt}
\begin{center}
\begin{minipage}{\boxwidth}
\begin{description}
\item[\sf Stage 1:] 
Each operational processor verifies if it is active: 
it scans its block looking for good segments, and organizes them in a list.
\item[\sf Stage 2:] 
Each active processor converts its list of good segments into a temporary 
tree.
\item[\sf Stage 3:] 
The active processors organize their active blocks into a list.
\item[\sf Stage 4:] 
The active processors convert the list of active blocks into a block 
tree.
\item[\sf Stage 5:] 
The active processors construct a processor tree from the list 
of active blocks.
\item[\sf Stage 6:] 
The active processors construct a memory tree. 
\item[\sf Stage 7:] 
A subset of the active processors run a decoding algorithm
to retrieve the input.
\end{description}
\end{minipage}
\FFF

\rule{\textwidth}{0.75pt}

\parbox{\captionwidth}{\caption{\label{fig-1}
Preprocessing in its seven main stages.}}
\end{center}
\end{figure}

After Stage~$6$, both the block tree and all the temporary trees 
are obsolete and may be discarded.
Next we discuss each stage in detail:

\noindent
{\bf Stage 1:} Each operational processor verifies if it is active.

To this end, it scans its assigned block of memory looking for good segments. 
As it moves through consecutive cells, it attempts to read each of them, 
checking to see if they are operational.
It also remembers the number of operational cells among the last $\alpha$
cells, not assigned to the previous good segments. 
Every time $\beta$ such cells are found, a new segment is added to the
list  of good segments, the first operational cell in a good segment storing 
the address of the next good segment, if any.
If no good segment is found then all the operational cells in a block are set 
to blank (all bits set to~$0$), 
otherwise the first such segment is the contact one, and
its first  cell stores a pointer to the next good segment if it exists,
or the null value if no such segment exists. 
In any case the value stored is distinct from  the  blank value.
If at least one good segment is found then the processor considers itself
active,  otherwise it is dormant.

 From now on, all the computations are performed by active processors
only, the dormant processors remaining completely idle.

\noindent
{\bf Stage 2:}
All the regular segments in every active block 
are arranged in a binary-search tree of height logarithmic in its size. 

The obtained tree is called this block's {\em temporary tree}.
This is also done sequentially in each block by the processor assigned 
to the block.
The in-ordering of regular segments in the temporary tree
is the same as their ordering by their physical addresses.
The total size of all the virtual storage in a block of the segments with 
smaller addresses than a given segment is this segment's {\em local offset}.
A regular segment stores the following in its structural cells: 
the offset and pointers to the children.
The local offsets can be computed and stored at each good segment in the list
while it is constructed.

A list of regular good segments is converted into a temporary tree
by the following {\em conversion procedure\/} that will be used also in 
future considerations.
It is partitioned into phases,
and can be performed either sequentially or in parallel,
depending on the number of available processors.
During this operation the list may be considered as storing trees, 
an item being a root of such a tree. 
In the beginning the trees are just single nodes.
A phase is performed as follows.
The list is partitioned into consecutive quadruples of nodes. 
Consider one of them.
It consists of tree $T_1$ followed by a single node $v_1$, followed by tree 
$T_2$, followed by a single node $v_2$, some possibly missing at the end of the
list.
The tree $T_1$ becomes the left child of $v_1$ and tree $T_2$ the right one,
then they are removed from the list, so that the tree with root $v_1$ is
immediately followed by $v_2$.
This is depicted in Figure~\ref{rys-1}.

%	figure : a phase of conversion

\begin{figure}[t]
\begin{center}
\includegraphics[width=\pagewidth]{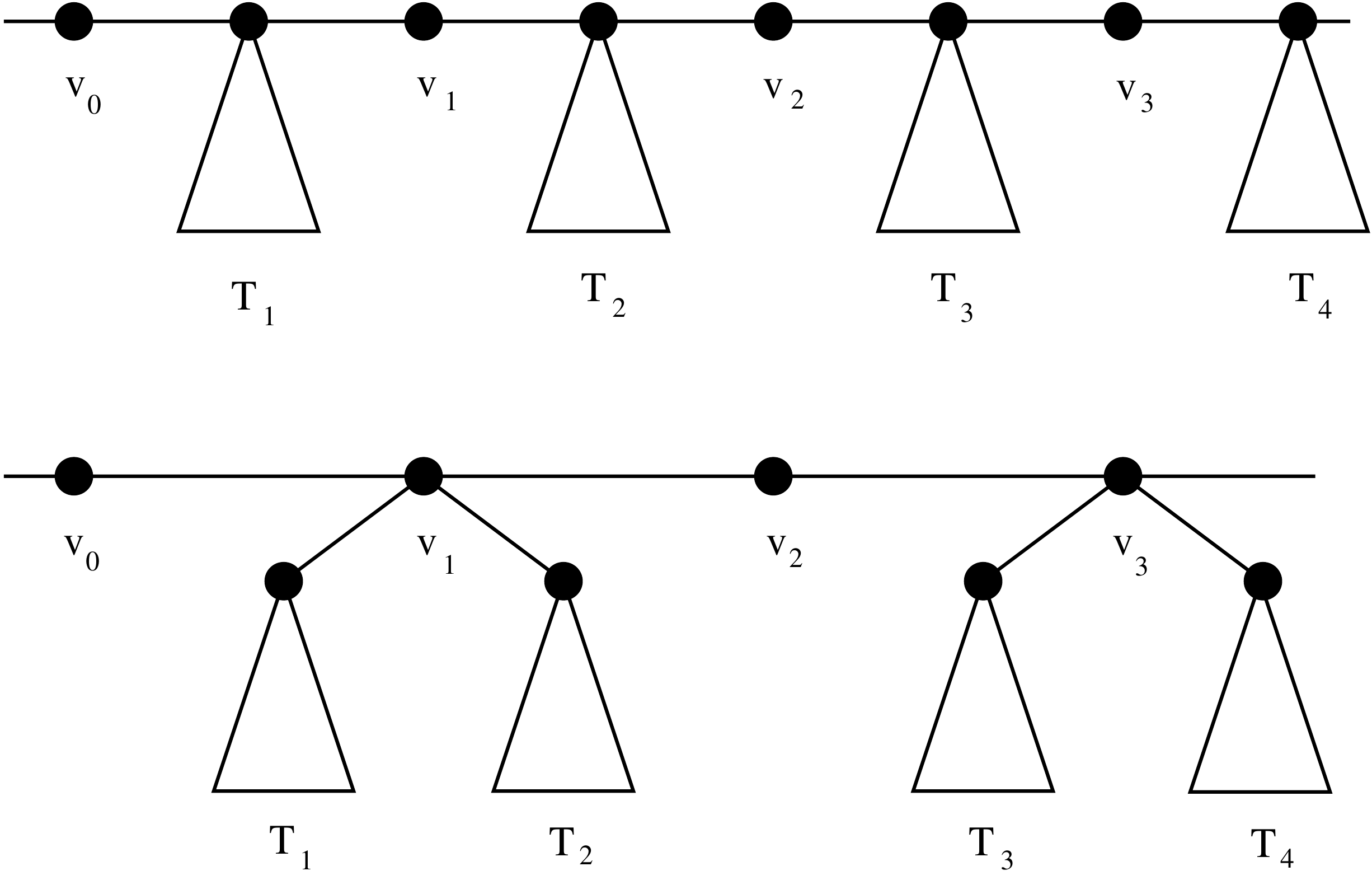}

~

\centerline{\parbox{\captionwidth}{\label{rys-1}
\caption{A phase of the conversion procedure. 
The upper list is transformed into the lower one.}}}
\end{center}
\end{figure}

After a number of phases, which is logarithmic in the length of the list,
just a single tree remains, also of depth logarithmic in the length of 
the list. 
This tree is the output of conversion.
In general, the time of conversion performed by just one processor
is proportional to the length of the input list, and if there is a processor
assigned to each item in the list then the task can be completed
in time logarithmic in the length of the list.
In the case of a temporary tree, its depth is $\cO(\log(m/n))$, and it is
obtained in time $\cO(m/n)$.

\noindent
{\bf Stage 3:}
A list of active blocks is built.

This list consists of contact segments.  
We use the following definitions.
A set $S$ of processors is {\em large\/} if it contains at least
$\frac{1}{2} (1-f_p-f_s) |S|$ active elements, where $|S|$ denotes 
the size of $S$.
A {\em partition tree\/} (PT in short) is
a full binary tree with nodes labeled by intervals of integers.
The root is labeled by $[1,n]$.
If a node $v$ is labeled by $[a,b]$ and $b>a+1$ then $v$ is an internal node,
its left child is labeled by $[a,\lfloor(a+b-1)/2\rfloor]$ and the right
child is labeled by $[\lceil(a+b)/2\rceil,b]$.
Leaves are labeled by intervals of unit length.
A processor with physical index~$i$ is {\em associated\/} with node 
$v$ if $i$ belongs to the label of~$v$.
The processors associated with a node of tree PT
are referred to as the {\em group\/} of this node.
A node is {\em large\/} if its group is large.
An active processor may {\em check\/} another processor $p$ for being active,
$p$ being operational or not.
This is done by scanning the block
assigned to $p$ and attempting to read each consecutive cell, until the
information that the block is active is found. This information is in the
first  operational non-blank cell in the block, which also happens to be
this  block's contact address.

The algorithm building the list works in phases corresponding to the levels 
of tree PT.
During a phase, processors are partitioned into groups 
associated with nodes at the corresponding level of the tree.
A group is either busy or not.
If a group is {\em busy\/} then:
\begin{enumerate}
\item
Each processor in the group is aware of this.
\item
The active processors in the group are organized in a list,
and each element of the list knows its {\em rank}, that is, the distance from
the head of the list.
\item
Each element knows the number of elements in the list.
\end{enumerate}
Groups correspond to leaves of PT during the first phase of the algorithm.
Processors in each such a group check each other, and if the group 
is large then the processors arrange themselves in a list, 
compute ranks and the group becomes busy.
In general, once a phase has been completed, the computation proceeds to the
next one, corresponding to the next level of PT towards the root.
Let us consider such a phase, it consists of three parts.
During the first part, the busy groups corresponding to nodes being left
children perform computations, while the other groups pause. 
The busy groups corresponding to right children perform computations
in the second part, 
but only if their left-sibling group was idle during part one.  
Consider the first part.
Each processor in a busy group $G_1$ is assigned a set of processors in the
sibling group $G_2$ as follows.
The processors in $G_2$ are partitioned into as many subsets as there are
active elements in $G_1$, the sizes of subsets are constant ($G_1$ is large)
and differ by at most one.
The active processor with rank $k$ in the list in $G_1$
is assigned the $k$th subset of group~$G_2$.
Each active processor of $G_1$ checks on the assigned subset and adds all
its active blocks to the list.
This completes part one.
The second part is similar, but the roles between sibling
groups are reversed.
For each pair of sibling groups, if at least one of the groups is busy then
after part two all the active processors in these
groups are connected in a combined list.
This list is processed during the third part.
The ranks and the size of the list are computed and broadcast to all the
elements of the list by the procedure of pointer jumping (see~\cite{JaJa92-book}).
If the size of the list is big enough for the group of the parent node
to become large then the group of the parent node becomes busy in the next 
phase.
To synchronize the process, we need to control the duration of the first 
two parts of a phase, when the processors check on the assigned elements 
of sibling groups.
This time period can be bounded by a common constant 
since all the busy groups are large.
\begin{lemma}
\label{l3}
It takes time $\cO\left(\log n\left(\frac{m}{n}+\log n\right)\right)$ 
to produce a list connecting all the active blocks.
\end{lemma}
\begin{proof}
Note that the root of PT is large by Lemma~\ref{l1}, and if a node of PT 
is large then one of its children is large.
Therefore there is a large leaf in PT such that the path from it to the root 
contains only large nodes.
This means that there is at least one busy group in each phase,
and the correctness follows by induction on the height of PT.

The number of phases is $\cO(\log n)$.
The first two parts of a phase take time $\cO(m/n)$ each, the third one takes 
time proportional to the logarithm of the size of the list, 
which is $\cO(\log n)$.
\end{proof}

\noindent
{\bf Stage 4:}
The active blocks are organized as a binary-search tree of logarithmic height.

This tree is called the {\em block tree}.
The in-ordering of active blocks in this tree, by their active indices,
is to be the same as their ordering by physical addresses.
The total size of virtual storage in all the blocks with active indices 
smaller than a given active block is called this block's {\em global offset}.
All the information about a block as a node in
the block tree is stored in the contact segment of this block.
Namely, this segment stores the following in its structural cells: 
the global offset, 
pointers to the children and to the parent in the block tree, and a
pointer to the root of the temporary tree.
A pointer to an active block is meant to be the address of the contact segment 
of this block.
The processors have already been organized into a list of contact segments.
The conversion procedure is applied to this list,
it takes $\cO(\log n)$ steps.
Then the global offsets of the nodes in the obtained tree are computed.
This is done in two sweeps through the levels of the tree.
The first one proceeds from the leaves towards the root. 
After it is finished then each node knows the size of virtual memory 
in all of its descendants.
The offsets are computed during the second sweep starting at the root and going
towards the leaves: a node needs to know the offset of its parent, and the size
of the virtual memory in the subtree of its left sibling, if there is any.

\noindent
{\bf Stage 5:}
A tree of processors is built.

This tree is called simply the {\em processor tree}.
We may assume that the number of active processors is a power of~$2$,
otherwise only $2^k$ active processors are considered, 
where the number $k$ is the largest one such that $2^k$ is upper bounded by the
number of active processors.
Every active processor has a number of cells in its contact 
segment assigned for the purpose of building the tree.
They are organized as two records of size suitable to store all the 
information of a node, each is said to be {\em owned\/} by the processor. 
One of them is designated to be a leaf in the tree, the other to be an 
internal node. 
When this stage starts then all the records to be internal nodes
are initialized to {\bf nil} values.
Each of them is called {\em spare\/} while holding such a value.
Since the list of contact segments has already been built, 
this also yields a list structure of the records assigned to be leaves.
Inductively, this list may be interpreted as having the following properties:
(i)~it connects trees, (ii)~the root of each tree in the list 
has a spare record.
Initially the trees are just singleton nodes.
The list is contracted in phases as follows.
The current list is partitioned into pairs of consecutive trees.
If $T_1$ and $T_2$ are such a pair then the spare record of $T_1$ becomes
the root of a new tree that replaces both $T_1$ and $T_2$ in the list,
it has $T_1$ as the left subtree and $T_2$ as the right one.
The spare record of $T_2$ becomes the spare one of the new tree.
This is performed for each such pair in parallel, which results in
contracting the length of the list by half.
After $k$ phases a single tree is obtained, which is the processor tree.

\noindent
{\bf Stage 6:}
A memory tree is constructed.

First the active processors are dispersed among good segments in a balanced
way, that is, they locate such segments that the distance, 
measured in terms of offsets, between any two consecutive ones is the 
same, and is $\Theta(m/n)$. 
This is achieved as follows.
A copy of the list of processors is first made, by way of making copies 
of the contents of the relevant cells in the contact segments of each 
processor.
It is then  dispatched  towards the leaves of the block trees, 
starting from the root.
While on its way, this list may be partitioned into smaller lists, they are 
called {\em traveling lists}.
These lists travel in a binary-search fashion, and when a leaf of the block
tree is encountered then they continue their way through the respective
temporary tree.
Let $L$ be such a list, and let it be at node $v$ 
of either the block tree or a temporary tree. 
The processor $P$ at the head of $L$ reads the offsets of both $v$ and its left
child, and calculates the smallest virtual address in the right child of~$v$.
This information is then broadcast to all the processors on list~$L$ by 
pointer jumping, in time proportional to the logarithm of the size of the list.
If a processor in the list is the first one with its address in the right
subtree of~$v$ then it cuts itself from the predecessor, if any, and becomes
the head of a new list consisting of the processors that follow it in the
original list.
This new list moves down to the right child of $v$, the remaining shortened
one goes to the left child.
Eventually the traveling lists become single nodes which traverse either
the block tree or temporary trees in a binary-search fashion.     

The second part of this stage starts from the processors 
dispersed among the items of the list of good segments in a balanced way, the
first processor at the first node of the list.
The good segments between that of processor $P_i$ and that of processor
$P_{i+1}$ are treated as the portion of segments assigned to
processor~$P_i$.
The sizes of such portions are $2^r$, for some $2^r=\Theta(m/n)$.
Every active processor can be assumed to know the number~$r$.
Each processor organizes $2^r$ cells of its initial portion of good segments 
into a complete binary-search tree. 
The list of active processors yields a list of these trees.
We apply the conversion procedure and obtain a complete 
binary-search tree, and this is the memory tree.

\begin{lemma}
\label{l-loglog}
It takes $\cO(\log n \log m)$ steps to build the memory tree.
\end{lemma}

\begin{proof}
Each partitioning of a traveling list takes $\cO(\log n)$ steps.
The length of the path from the root of the block tree to a leaf of a temporary
tree in a block is $\cO(\log m)$.
\end{proof}

\noindent
{\bf Stage 7:}
A decoding algorithm is run.

By Lemma~\ref{l1}, the number of active processors is at least 
$\gamma n$, where $\gamma=\frac{1-(f_p+f_3)}{2}$.  
All of them are arranged in a list corresponding to the list of blocks.
We want to provide the first $\gamma n$ processors in this list
with the whole input distributed among all~$n$ processors, 
active or not.
To this end we use a variant of information dispersal method, 
based on an idea similar to that 
underlying the Reed-Solomon error-correcting codes (\cite{Lyuu93-book,Preparata89}). 
The original input is first partitioned into $\gamma n$ {\em input
strings\/} $u_1$, $\ldots$, $u_{\gamma n}$, each comprised of some $t$ bits.
These input strings will be interpreted as elements of the finite field
$\mF=\text{GF}(2^t)$.
We take the smallest~$t$ such that the inequality $2^t>\gamma n$
holds.
We assume that the input is encoded as a sequence of $n$ packets 
$\langle 1,v_1\rangle$, 
$\langle 2,v_2\rangle$, $\ldots$, $\langle n,v_n\rangle$, and packet 
$\langle i,v_i \rangle$ is delivered to the processor with the physical
index~$i$, faulty or not.
Each packet $\langle i,v_i \rangle$ is a pair of {\em identifier\/ $i$} and
{\em value\/ $v_i$}.
The values $v_i$ are sequences of $t$ bits each, 
they are interpreted as elements of field $\mF$, 
with the arithmetic of operations on them defined as in $\mF$.
The values satisfy $v_i=g(\omega^i)$, 
where the following holds:
\begin{enumerate}
\item
$\omega$ is a primitive element of $\mF$;
\item
$g(x)$ is the following polynomial in $\mF[x]$:
\[
g(x)=u_1+u_2x+\ldots+u_{\gamma n}x^{\gamma n -1}\ .
\]
\end{enumerate}
Note that in order to encode the input we need to evaluate a polynomial at a
subset of $\mF$ of size $\Theta(2^t)$.
The task of encoding is equivalent to computing the Discrete Fourier Transform
and can be implemented efficiently by the FFT algorithm.
Suppose that $\langle i_1,v_{i_1}\rangle$, $\ldots$, 
$\langle i_\gamma, v_{i_{\gamma n}}\rangle$ are the input packets stored by
the initial segment of $\gamma n$ active processors in the list (in any order).
This gives $\gamma n$ values at distinct points, which are determined by the
identifiers of packets of a polynomial of degree $\gamma n -1$.
The task of retrieving the input $u_1$, $\ldots$, $u_{\gamma n}$ is equivalent
to obtaining the coefficients of a polynomial from its values, that is, to
interpolating the polynomial.
There is an algorithm for this problem that runs in time $\cO(\log^3n )$ on an
EREW PRAM with $n$ processors (see~\cite{JaJa92-book} and the references therein).
The algorithm computes the Lagrange interpolation formula and is reduced to
polynomial evaluation at $\cO(n)$ points and then to the FFT algorithm.
This algorithm specialized to our problem can be implemented to run faster,
even on the butterfly
(see \cite{Leighton1992} for the description and properties of the butterfly).
By a {\em normal butterfly algorithm\/} we mean an algorithm in which a step 
of computation is performed by the nodes on one level, and the consecutive 
levels are used in a cyclic fashion.
\begin{lemma}
\label{decoding}
Decoding of the whole input from $\gamma n$
packet values can be performed on the $\cO(\log n)$-dimensional butterfly in 
time $\cO(\log^2 n)$ by a normal algorithm.
\end{lemma}
\begin{proof}
A general interpolation algorithm resorts to an algorithm of evaluating a
polynomial at $\cO(n)$ points, which runs in time $\cO(\log^2n)$. 
This algorithm is replaced by an algorithm to evaluate polynomials at the 
powers of a primitive element of a finite field.
Note that we need to evaluate the polynomials at some bounded fraction of 
{\em all\/} the elements of the field, which can be performed in time 
$ \cO(\log n)$ by the FFT algorithm; in this way we gain a $\log n$-factor.
The communication of processors is that needed for the FFT algorithm and of a
full-binary-tree pattern, and can be implemented on the butterfly as a normal
algorithm.
\end{proof}
We can adapt this algorithm, due to the following:
\begin{lemma}
\label{butterfly}
A normal butterfly algorithm can be implemented on a list of active processors
with delay $\cO(1)$, provided the size of the list is at least equal to the
number of nodes in one level of the butterfly.
\end{lemma}
\begin{proof}
We need a two-directional cyclic list of length equal exactly to the size of a
level of the butterfly, but this has been taken care of during preprocessing.
Each processor simulates a row of the butterfly.
The connections to other rows can be obtained dynamically by pointer jumping
along the links of the list of active processors.
A processor chooses the particular link depending on the bit representation of
the row number.
\end{proof}

\begin{lemma}
\label{lem-retrieve}
The input available to the active processors 
can be retrieved in time $\cO(\log^2 n)$ on an EREW PRAM. 
\end{lemma}

\begin{proof}
Time evaluation of Stage~7 follows from Lemmas~\ref{decoding}
and~\ref{butterfly}. 
\end{proof}

\begin{theorem}
\label{t-1}
The preprocessing performed by a normal PRAM with $n$ processors and $m$ 
memory cells takes time 
$\cO\left(\log n\left(\frac{m}{n}+\log n\right)\right)$.
\end{theorem}
% Preprocessing designates $\Omega(n)$ operational processors 
% as active and creates $\Omega(m)$ virtual storage cells.
\begin{proof}
Stage~1 takes time $\cO(m/n)$. 
The time needed to perform Stage~2 is proportional
to the number of nodes in temporary trees, which is~$\cO(m/n)$.
Lemma~\ref{l3} gives the bound on the time to perform Stage~3, which dominates
the time bound of the preprocessing.
Both Stage~4 and Stage~5 take time~$\cO(\log n)$. 
By Lemma~\ref{l-loglog}, Stage~6 takes time 
$\cO(\log n \log m)
=\cO(\log n (\log\frac{m}{n}+\log n))$.
The input is retrieved in time $\cO(\log^2 n)$ per $\cO(1)$ memory cells
for each processor by Lemma~\ref{lem-retrieve}.
\end{proof}

\section{Simulation Proper}
\label{simulation}

In this section we explain a method to simulate any program of a fully
reliable PRAM on the same variant (EREW,CREW,CRCW, respectively) 
as the simulator, when the simulator may have faulty components.
Its slowdown is $\cO(\log m)$ per simulated step, this overhead is
asymptotically equal to the time needed to locate a physical address 
of a memory cell given its virtual address.

The simulation relies on the results of preprocessing, 
as described in the preceding section.
The preprocessing terminates with an initial segment of active processors
storing the input. 
The input given to the simulator consists of a program $\cal P$ to execute 
and data $\cal D$ which are to be input for $\cal P$.
We assume that $\cal P$ is written in a low-level machine-like language.

We start with a general scheme of the simulation.
The simulation proper starts with the active processors 
writing $\cal P$ and $\cal D$ into a contiguous initial segment 
of the virtual shared memory to make it available to all the simulated 
processors.
Then $\cal P$ is executed step by step.
Each of the active processors simulates $\cO(1)$ processors.
Active processors have their local memories partitioned into regions.
One region is needed for the purpose of running the 
simulating algorithm, the remaining ones are used 
as local memories of the simulated processors. 
Each simulated processor needs only to perform specific machine 
instructions that are known to the respective active processor 
of the simulator.
This means in particular that there are no additional local-memory 
requirements of the simulated processors except for those required by $\cal P$.
We assume that local memories of processors of the simulating normal PRAM
machine are sufficiently large to simulate $n$ processors in total.
Moreover we assume that the repertoire of machine instructions of the simulator
and the simulated machine are exactly the same.
Hence all the machine instructions, with the exception of a shared memory 
access, can be performed by the simulator in constant time.

What remains to be explained is how an access to the virtual shared memory is
implemented.
Let us consider a single instance of access, and assume that each active 
processor already knows the virtual address to find, if any.
If the simulator is either a CREW or a CRCW PRAM then the respective
physical address can be found with the help of only the memory tree.
To this end, each active processor traverses this tree in a binary-search
manner until eventually it reaches a leaf with the address.
The traversal requires a concurrent-read capability, since all the active
processors start from reading the contents of the root of the memory tree,
and then a number of them may read simultaneously the contents 
of other nodes of the tree.
Care needs to be taken to have all the actual operations of reading/writing
to virtual address performed in the same time step on the memory cells
storing the contents of the virtual cells, to preserve the semantics 
of concurrent access of the simulator. 
This is easy since the memory tree is perfectly balanced, 
so the times to locate any leaves are the same.
The total time of an operation of reading/writing to a virtual memory cell
is of the order of the height of the memory tree,
which is $\cO(\log m)$.

In what follows we present a solution in the case of the EREW PRAM.
The search for physical addresses is in two stages: {\em construction\/}
which is followed by {\em traversal}.
During a construction stage the active processors are arranged in a special
tree  that we call {\em organized}.
During a traversal, the root of the organized tree
descends down the memory tree, starting from its root,
and in the process the tree disintegrates into smaller trees.
These trees keep moving down the memory tree and splitting into smaller ones
until eventually single nodes/processors remain.
This way of traversal of the memory tree guarantees that no 
concurrent-reading is ever performed until the leaves are reached.
An organized tree has the property that each of its subtrees is also 
organized.
The structure of an organized tree is such that the traversal is fairly
direct.
This structure is implicitly determined by the way in which the tree 
is to be used to locate virtual addresses in the memory tree.

Each active processor is  associated with one leaf of an organized tree,
and with at most one internal node.
We say that the processor owns these nodes, similarly as in the description
of Stage~5 of preprocessing.
More precisely, all except for one processors owning the leaves also own one
internal node.
The one spare node is known by the processor owning the root, this knowledge 
is used during construction.
An internal node $v$ stores both the smallest and the largest
among the virtual addresses of the processors at the leaves of the subtree
rooted at $v$, we call them {\em leftmost\/} and {\em rightmost addresses\/}
of~$v$.
The operation of construction is performed repeatedly for each operation 
of virtual memory access. 
The current organized tree is discarded after a traversal has been completed.

\noindent
{\bf Stage of traversal.}\ 
Organized trees traverse the memory tree in a binary-search fashion. 
Consider such a tree $T$ during traversal, which is at a node $v$ 
of the memory tree.
The processor that owns the root of $T$ can decide in constant time
if the tree splits here or not.
In the former case it moves down the tree to one of the children of~$v$.
This happens if all the virtual addresses searched by the processors of $T$
are in one of the subtrees of node~$v$.
Otherwise the tree is split: the left subtree of $T$ heads for the left child
of $v$ and the right one for the other child. 
The root of $T$ is discarded.

It is the structure of $T$ which makes it possible to make this decision in a
constant time and to guarantee correctness.
If $v$ is a node of an organized tree then 
the leftmost addresses (or the rightmost ones) of the children of~$v$
determine the level of the memory tree at which the tree with root at $v$
will split while traversing the memory tree to carry the processors at their
leaves to their virtual addresses.
This level can be computed in constant time 
because the memory tree is perfectly balanced.

\noindent
{\bf Stage of construction.}
The following is a detailed recipe how to construct an organized tree for a
given memory access operation.
The active processors start from the leaves of the processor tree and 
move towards the root.
The computation proceeds in phases corresponding to the levels of this tree.
A phase starts with organized trees associated with all the nodes on a level
of the processor tree.
Every two trees associated with siblings are merged into one tree which is then
moved to the parent.

We need to define one more notion, that of the {\em binary path\/} 
of an organized tree $T$.
If $T$ is a single node, that is, a single processor, 
then let $v$ denote the node of the memory tree
storing the virtual address needed by the processor.
Otherwise let $v$ be the node of the memory tree where tree $T$ splits for the
first time while traversing the memory tree.
The path from the root of the memory tree down to $v$ can be encoded in binary:
start with an empty sequence, then traverse the links towards $v$ appending
$0/1$ if moving to a left/right child, respectively.
Such binary paths can be compared by the lexicographic ordering: $x$ precedes
$y$ if $x$ has $0$ at the first position on which $x$ and $y$ differ,
or if $x$ is a prefix of~$y$.
During construction, each organized tree stores its binary path at the root.
If an organized tree is a single node/processor $P_i$ then the binary 
path of $P_i$ can be computed in constant time: 
it is the binary representation of the number
equal to the virtual address that $P_i$ needs to locate.
If two organized trees with binary paths $r_1$ and $r_2$ are merged then the
resulting organized tree has the longest common prefix of $r_1$ and $r_2$ 
as its binary path.

A procedure to merge two organized trees $T_1$ and $T_2$ is defined
recursively as follows.
Let $r_1$ and $r_2$ be the binary paths of $T_1$ and $T_2$, respectively.
The spare record of $T_1$ is designated as the root, and the spare record of
$T_2$ is designated as the spare record of the merged tree.

\noindent
{\bf Case 1:} None among $r_1$ and $r_2$ is a prefix of the other.

Trees $T_1$ and $T_2$ are made children of the root in such order that the
binary path of the left subtree precedes the binary path of the right one.

\noindent
{\bf Case 2:} Paths $r_1$ and $r_2$ are equal.

The left subtrees of $T_1$ and $T_2$ are merged recursively and then attached
as the left subtree of the root, similarly the right subtrees.

\noindent
{\bf Case 3:} One of the paths $r_1$ and $r_2$ is a proper prefix of the other.

%	figure : merging

\begin{figure}[t]
\begin{center}
\includegraphics[width=\pagewidth]{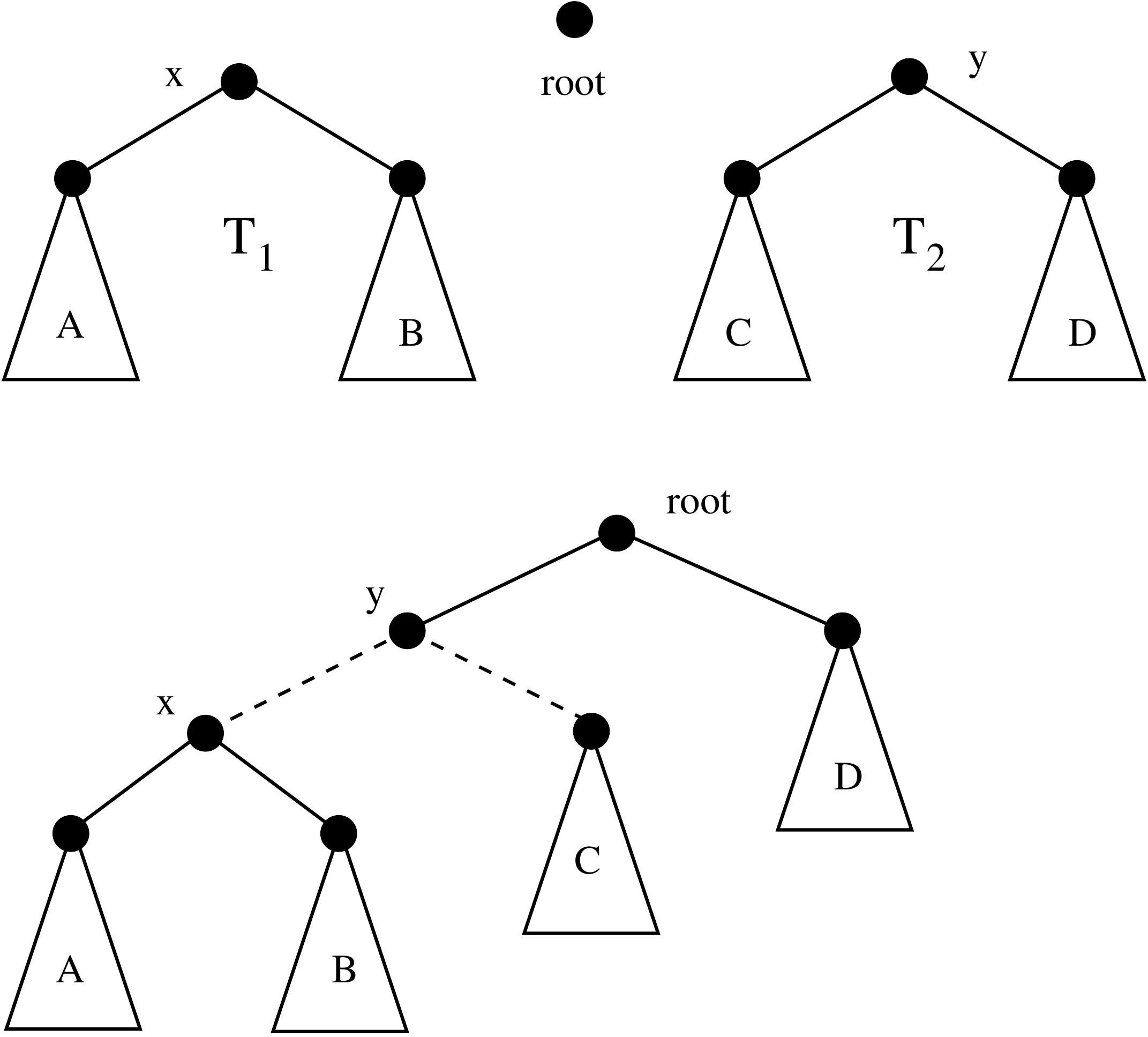}

~

\centerline{\parbox{\captionwidth}{\caption{\label{rys-2} 
Merging organized trees, as in Case 3. 
Dotted lines indicate that the tree rooted at $x$ and tree $C$ still need to be merged recursively.
}}}
\end{center}
\end{figure}

Suppose that $r_1$ is a prefix of $r_2$, and also that the first symbol after
prefix $r_1$ in $r_2$ is $0$, the other cases are symmetric.
The right subtree of $T_2$ becomes the right subtree of the root.
The root of $T_2$ is made the left child of the new root. The left subtree 
of the root is obtained by merging $T_1$ with the left subtree of $T_2$.
This is depicted in Figure~\ref{rys-2}.

This recursive procedure is performed by going from the root 
towards the leaves, taking time $\cO(1)$ at each level.
Note that it can be executed in a pipelined fashion:
once roots have rearranged links and a new root knows its children,
the root of the merged tree can be moved up 
to the parent in the processor tree.
The delay is $\cO(1)$ because each active processor serves at most 
two nodes in both the organized and the processor trees.

\begin{theorem}
The memory access for a single step can be performed in time $\cO(\log m)$.
\end{theorem}

\begin{proof}
The cases of CREW and CRCW PRAMs have been explained before,
we need to consider only the EREW PRAM and the solution by way of organized
trees.
Merging all of the trees would take time $\cO(\log^2 m)$ if performed
without pipelining, which saves a factor of $\cO(\log m)$.
The actual traversal of the memory tree takes time $\cO(\log m)$.
Hence an organized tree can be built in time $\cO(\log m)$.
The task of locating virtual addresses is performed in a binary search fashion
and takes also time $\cO(\log m)$.
\end{proof}

\vspace{2ex}

\noindent
{\bf Acknowledgement.}
The authors thank Krzysztof Diks for discussions of the related fault-tolerance issues, 
and to Piotr Indyk for sharing his insights on the Discrete Fourier Transform and information dispersal.

%: bibliography

\bibliographystyle{plain}

\bibliography{bogdan,books,other,parallel}

\end{document}